\documentclass[11pt]{amsart}

\usepackage[utf8]{inputenc}
\usepackage[nospace,noadjust]{cite}
\usepackage{relsize}
\usepackage{color,verbatim}
\usepackage[dvipsnames]{xcolor}
\usepackage{tikz-cd}
\usetikzlibrary{calc,fit,matrix,arrows,automata,positioning}
\usetikzlibrary{decorations.pathreplacing,angles,quotes}
\usepackage{colortbl}
\usepackage{hyperref, enumitem}
\hypersetup{
  colorlinks   = true,
  urlcolor     = blue,
  linkcolor    = Purple,
  citecolor   = red
}
\usepackage{float}
\usepackage{amsmath,amsthm,amssymb,mathdots}
\usepackage{mathtools}
\usepackage{mathrsfs}
\usepackage[margin=2.5cm]{geometry}
\usepackage{stmaryrd}
\usepackage{array}
\newcolumntype{x}[1]{>{\centering\arraybackslash\hspace{0pt}}p{#1}}
\theoremstyle{definition}
\newtheorem{theorem}{Theorem}[section]
\newtheorem{definition}[theorem]{{{Definition}}}
\newtheorem{example}[theorem]{{{Example}}}
\newtheorem{remark}[theorem]{{{Remark}}}
\newtheorem{corollary}[theorem]{{{Corollary}}}
\newtheorem{proposition}[theorem]{{{Proposition}}}
\newtheorem{lemma}[theorem]{{{Lemma}}}

\newcommand{\numberset}{\mathbb}
\newcommand{\N}{\numberset{N}}

\newcommand{\F}{\numberset{F}}

\newcommand{\mC}{\mathcal{C}}

\newcommand{\mG}{\mathcal{G}}

\newcommand{\mI}{\mathcal{I}}

\newcommand{\mU}{\mathcal{U}}
\newcommand{\mV}{\mathcal{V}}
\newcommand{\mO}{\mathcal{O}}

\newcommand{\Fq}{\F_q}
\newcommand{\Fm}{\F_{q^m}}
\newcommand{\Fmn}{\Fm^n}

\newcommand{\Fmk}{[n,k]_{q^m/q}}
\newcommand{\ba}{\bar\alpha}
\newcommand{\st}{\,:\,}

\DeclareMathOperator{\rk}{rk}
\DeclareMathOperator{\GL}{GL}
\DeclareMathOperator{\PG}{PG}

\DeclareMathOperator{\T}{T}

\usepackage{graphicx}

\title{A geometric invariant of linear rank-metric codes}
\usepackage[foot]{amsaddr}
\author{Valentina Astore$^{1,2}$}
\author{Martino Borello$^{3,1}$}
\author{Marco Calderini$^4$}
\author{Flavio Salizzoni$^5$}
\address{$^1$INRIA, France}
\address{$^2$LIX, École polytechnique, Institut Polytechnique de Paris, France}
\address{$^3$Universit\'e Paris 8, Laboratoire de G\'eom\'etrie, Analyse et Applications, LAGA, Universit\'e Sorbonne Paris Nord, CNRS, UMR 7539, France.}
\address{$^4$University of Trento, Italy.}
\address{$^5$MPI-MiS, Leipzig, Germany.}

\email{valentina.astore@inria.fr}
\email{martino.borello@univ-paris8.fr} \email{marco.calderini@unitn.it} 
\email{flavio.salizzoni@mis.mpg.de}
\thanks{V.~A. is funded by AID (French \emph{Agence de l'innovation de défense}).}
\thanks{M.~B. is partially supported by the ANR-21-CE39-0009 - BARRACUDA (French \emph{Agence Nationale de la Recherche}) and by PHC GALILEE 2024 Projet n° 50424WM - ``Algebraic and Geometric methods in coding theory''.}
\thanks{The research of M.~C. was partially supported by the Italian National Group for Algebraic and Geometric Structures and their Applications (GNSAGA - INdAM) and by the Italian Ministry of University and Research with the project PRIN 2022RFAZCJ}
\thanks{F.~S. is supported by the P500PT\textunderscore  222344 SNSF project and by the 2023-28 Germaine de Staël project}

\begin{document}

\begin{abstract}
Rank-metric codes have been a central topic in coding theory due to their theoretical and practical significance, with applications in network coding, distributed storage, crisscross error correction, and post-quantum cryptography. Recent research has focused on constructing new families of rank-metric codes with distinct algebraic structures, emphasizing the importance of invariants for distinguishing these codes from known families and from random ones. In this paper, we introduce a novel geometric invariant for linear rank-metric codes, inspired by the Schur product used in the Hamming metric. By examining the sequence of dimensions of Schur powers of the extended Hamming code associated with a linear code, we demonstrate its ability to differentiate Gabidulin codes from random ones. From a geometric perspective, this approach investigates the vanishing ideal of the linear set corresponding to the rank-metric code.
\end{abstract}

\maketitle

\textbf{MSC Classification Codes:} 11T71, 51E20, 94B27. \smallskip

\textbf{Key words:} Rank-metric codes, Schur product, Overbeck's distinguisher, generalized Gabidulin codes, Castelnuovo-Mumford regularity.

\section*{Introduction}
Rank-metric codes have been a topical subject in coding theory since their introduction in 1978. While their initial development was motivated by theoretical reasons~\cite{Delsarte,Gabidulin}, the interest in these codes stems from practical applications such as network coding~\cite{silva2008rank}, distributed data storage~\cite{silberstein2012error}, crisscross error correction~\cite{roth1991maximum}, and code-based cryptography.  Interestingly, the first code-based cryptosystem utilizing Gabidulin codes was proposed as early as 1991~\cite{gabidulin1991ideals}. However, it was not until the recent push for post-quantum cryptography that Gabidulin codes gained significant prominence in this area. Although no public-key cryptosystem based on rank-metric codes has made it past the second round of the NIST Post-Quantum Cryptography Standardization process, NIST believes that rank-based cryptography should continue to be researched, as rank-metric cryptosystems offer a nice alternative to traditional Hamming metric codes with comparable bandwidth~\cite{alagic2020status}. Moreover, in the second call for Digital Signature Schemes launched by NIST, the rank-metric based schemes Mirath (merge of MIRA \cite{aragon2023miradigitalsignaturescheme} and MiRitH \cite{cryptoeprint:2023/1666}) and RYDE \cite{bidoux2023rydedigitalsignaturescheme} have currently made it through the first round and are still competing in the selection process.

Driven by these applications, in recent years there has been a growing interest in building new families of optimal codes with different algebraic structures, especially since it has been shown that most MRD codes are not Gabidulin codes when sufficiently large field extension degrees are considered ~\cite{neri2018genericity}. A first explicit construction for different MRD codes, called \emph{twisted Gabidulin codes}, was presented in ~\cite{sheekey2015new}. Other constructions can be found, among the others, in ~\cite{csajbok2018new,lunardon2018generalized,HorlemannTrautmann-Marshall,sheekey2020new}. For a more complete list of results we refer to~\cite{sheekey201913}. 

A key step in constructing new codes is to prove that they are not equivalent to already known families. Although determining whether two $\Fm$-linear rank-metric codes are equivalent or not can be done in polynomial time~\cite{couvreur2021hardness}, there is much interest in constructing invariants that allow to complete this task in an easy way.
Moreover, in a code-based cryptographic system whose security relies on the difficulty of detecting the code’s algebraic structure, identifying an invariant that discriminates structured codes from random ones results in mining any security proof of the scheme and could potentially lead to an attack. However, the literature on this topic is currently still limited and only few useful invariants are available. For instance, in~\cite{HorlemannTrautmann-Marshall} the authors introduced an invariant for generalized Gabidulin codes, which is based on the dimension of the intersection of the code with itself under some field automorphism. This was investigated further in~\cite{giuzzi2019identifiers} and generalized in~\cite{neri2020equivalence}.

Even in the scenario of the Hamming metric, similar problems arise, and in this case, the Schur product plays a fundamental role. Inspired by~\cite{faugere2013distinguisher}, in~\cite{Corbella-Pellikaan} it was proven that the dimension of the Schur product of a linear code (in the Hamming metric) with itself allows us to differentiate between a generalized Reed Solomon code from a random one. This invariant was then used successfully in~\cite{couvreur2014distinguisher} to construct a key-recovery attack on a code-based cryptographic scheme.

In this paper, we introduce a geometric invariant for linear rank-metric codes which is inspired and similar to the dimension of the Schur products. In the case of rank-metric codes, just considering the Schur powers is not enough and, moreover, is not rank-equivalence-invariant. To overcome this problem, it is necessary to consider the extended Hamming code, introduced in~\cite{Alfarano-Borello}, and examine its Schur products. In particular, we show that the sequence of the dimensions of the Schur powers of the extended code associated with a linear code distinguishes a Gabidulin code from a random one. From a geometric point of view, this corresponds to investigate the vanishing ideal of a linear set corresponding to the rank-metric code. In particular, the behavior of forms of a certain degree distinguishes (generalized) Gabidulin codes from random ones. In the last part of the paper we inquired which properties of the code can be determined from the dimension sequence. This leads also to an investigation of zeros of forms on linear sets.
\bigskip

\noindent \textbf{Outline:}
The structure of this paper is as follows. Section \ref{section:preliminaries} presents the preliminaries on rank-metric codes, including an overview of known invariants in the rank-metric context, an introduction to the Schur product in the Hamming metric, and the geometric interpretation of rank-metric codes. In Section \ref{section:dimension sequence of rank-metric codes}, we define and analyse the dimension sequence of codes in the rank metric. In particular, we study the behavior of Gabidulin codes and compare it with the general case of random codes. Section \ref{section:Evaluation of forms on linear sets} explores the geometric properties that support the results presented in the previous section. Specifically, we demonstrate how the findings from Section \ref{section:dimension sequence of rank-metric codes} relate to the evaluation of forms on linear sets. Finally, Section \ref{section:number of zeros of a form} is dedicated to study the zero locus of the forms introduced in Section \ref{section:Evaluation of forms on linear sets}.
\bigskip

\section{Background}\label{section:preliminaries}
Throughout this paper, let $q$ be a prime power, $\Fq$ be the finite field of order $q$, and $m,n$ be two positive integers. 

\begin{definition}
    Let $v=(v_1,\ldots,v_n),w=(w_1,\ldots,w_n)\in\Fm$. The rank metric over $\Fm$ measures the distance between $v$ and $w$ as   
    \[\rk (v,w):=\dim_{\F_q}\langle v_1-w_1,\ldots,v_n-w_n \rangle_{\F_q}.\]
    A (linear) \emph{rank-metric code} $\mC$ is an $\F_{q^m}$-linear subspace of $\F_{q^m}^n$ endowed with the rank metric. If $k$ is its dimension and $d$ is its minimum rank distance, we say that $\mC$ is an $[n,k,d]_{q^m/q}$ code.
\end{definition}

\noindent A \emph{generator matrix} $G$ for a code $\mC$ is a matrix whose rows form a basis of $\mC$. If the $\F_q$-dimension of the columns of $G$ is equal to $n$, then we say that $\mC$ is \emph{nondegenerate} (see~\cite[Definition 3.1. and Proposition 3.2.]{Alfarano-Borello} for equivalent definitions). Two rank-metric codes $\mC$ and $\mC'$ are \emph{equivalent} if there exists an $\F_{q^m}$-linear \emph{isometry} (that is a map $\varphi_{\alpha,A}:v\mapsto \alpha v A$, with $\alpha\in \F_{q^m}^\ast$ and $A\in \GL_n(\F_q)$) sending $\mC$ to $\mC'$. We refer to~\cite{gorla2021rank} for a complete introduction on this topic. 

\medskip

\subsection{Maximum rank distance codes} We introduce now the main objects of the paper.

\begin{theorem}[{\cite[Section 2]{Gabidulin}}]
    Let $\mC\subseteq\Fmn$ be an $[n,k,d]_{q^m/q}$ code with $n\leq m$. Then \[k\leq n-d+1.\]
\end{theorem}

\noindent Codes attaining the previous bound are called \emph{maximum rank distance $($MRD$)$ codes}. Any linear MRD code $\mathcal{C}\subseteq\F_{q^m}^n$ of dimension $k$ has a generator matrix $G\in \F_{q^m}^{k\times n}$ in systematic form, i.e.,
$$G= [\, I_k \, | \, X\, ],$$ where all entries of $X$ are elements of $\F_{q^m}\setminus\F_q$~\cite[Lemma 5.3]{HorlemannTrautmann-Marshall}. It was shown in~\cite[Theorem 4.6]{neri2018genericity} that when large field extension degrees are considered, then a randomly chosen systematic generator matrix defines an MRD code with high probability. \\ In this paper, we focus on a specific family of MRD codes, originally introduced independently by Delsarte in~\cite{Delsarte} and by Gabidulin in~\cite{Gabidulin}, and later generalized in~\cite{kshevetskiy2005new}.
\begin{definition}\label{definition:Gabidulin}
    Let $\alpha_1,\dots,\alpha_n\in\Fm$ be linearly independent elements over $\Fq$ and $s\in\N$ such that $\mathrm{gcd}(s,m)=1$. A \emph{generalized Gabidulin code $\mG_{s,k}(\alpha_1,\dots,\alpha_n)$ with parameter $s$ and of dimension $k$} over $\Fmn$ is the rank-metric code whose generator matrix is
    \[\begin{pmatrix}
    \alpha_1 & \alpha_2 & \cdots & \alpha_n \\ \alpha_1^{[s]} & \alpha_2^{[s]} & \cdots & \alpha_n^{[s]} \\ \vdots & \vdots & & \vdots \\ \alpha_1^{[s(k-1)]} & \alpha_2^{[s(k-1)]} & \cdots & \alpha_n^{[s(k-1)]}
\end{pmatrix},\]
where, for a nonnegative integer $i$, we write $[i]$ to mean $q^i$. In particular, when $s=1$, we will simply say that $\mG_k=\mG_{1,k}(\alpha_1,\dots,\alpha_n)$ is a \emph{Gabidulin code of dimension $k$}.
\end{definition}

\noindent In the following, we will frequently compare structured MRD codes with random ones. Thus, we conclude this section by giving a formal definition of this notion. For us, a \emph{random code} is a random variable that selects with uniform probability a matrix in systematic form and takes the code generated by it.

\medskip

\subsection{Known invariants in the rank metric}
Hereafter, for a positive integer $s$ and a given matrix (or vector) $X$, we denote by $X^{[s]}$ the matrix (or vector) obtained by raising all entries of $X$ to the power $q^s$. \\ The following results present useful characterizations of generalized Gabidulin codes.
\begin{theorem}[{\cite[Theorem 4.8]{HorlemannTrautmann-Marshall}}]
    Let $\mC\subseteq\Fmn$ be a linear MRD code of dimension $k<n$ and let $\mC^{[s]}
    :=\{c^{[s]}\st c\in\mC\}$ for some $s\in\{1,\dots,m-1\}$ with $\gcd(s, m) = 1$. Then, $\dim(\mC\cap\mC^{[s]}) = k-1$ if and only if $\mathcal{C}$ is a generalized Gabidulin code with parameter $s$.
\end{theorem}
\begin{lemma}[{\cite[Lemma 3.3]{neri2018genericity}}]\label{lm:rank1}
Let $\mathcal{C}\subseteq\F_{q^m}^n$  be a linear MRD code of dimension $k$ with generator matrix $ [\,I_k\,|\, X\,]$, and $s\in\{1,\dots,m-1\}$ with $\gcd(s, m) = 1$. Then, $\mathcal{C}$ is a generalized Gabidulin code with parameter $s$ if and only if $\mathrm{rk}(X^{[s]} - X) = 1$.
\end{lemma}

\noindent Following the notation introduced in~\cite{CouvreurZappatore}, for a positive integer $i$, we will denote the $i$-th $q$-sum of a linear rank-metric code $\mC\subseteq\Fmn$ by \[\Lambda_i(\mC):=\mC+\mC^{[1]}+\dots+\mC^{[i]}.\]
As firstly introduced by Overbeck in his attack for the Gabidulin version of McEliece cryptosystem
(GPT)~\cite{Overbeck}, if the degree of the extension field $m$ is large enough, then the first $q$-sum distinguishes a Gabidulin code from a random one. 
\begin{theorem}[{\cite[Proposition 3]{couvreur2014distinguisher},\cite[Proposition 1]{coggia2020security}}]\label{theorem:randomqsum}
    Let $\mG_k,\mC\subseteq\Fmn$ respectively be a Gabidulin and a random $k$-dimensional code. Then, for any positive integer $i<n-k$, it holds that \[\dim\big(\Lambda_i(\mG_k)\big)=k+i<k(i+1)=\dim\big(\Lambda_i(\mC)\big),\] where for any nonnegative integer $a$, $\Pr\big(\dim\big(\Lambda_i(\mC)\big)\leq k(i+1) - a\big) = \mO(q^{-ma})$. 
\end{theorem}

\noindent Further investigations and generalizations of these results can be found, e.g., in~\cite{giuzzi2019identifiers, HorlemannTrautmann-Marshall, neri2018genericity, CouvreurZappatore, couvreur2014distinguisher, hormann2022distinguishing}. \medskip

\subsection{Schur product in the Hamming metric}

We denote by $\ast$ the standard component-wise product in $\F_{q^m}^n$, i.e., for $v=(v_1,\ldots,v_n),\,w=(w_1,\ldots,w_n)\in\F_{q^m}^n$, \[v\ast w:=(v_1w_1,\ldots,v_nw_n).\]
Let $\mC_1,\mC_2\subseteq\F_{q^m}^n$ be two linear codes. The \emph{Schur product} $\star$ of $\mC_1$ and $\mC_2$ is defined as \[\mC_1\star\mC_2:=\langle c_1\ast c_2\st c_1\in\mC_1,c_2\in\mC_2\rangle_{\F_{q^m}}.\]    
It is not hard to see that the set of linear codes contained in $\F_{q^m}^n$, together with the operations $+$ and $\star$, is a commutative semiring ordered under inclusion. \\ Let $\mC\subseteq\F_{q^m}^n$ be a linear code and let $\mC^{(0)}=\langle(1,\dots,1)\rangle_{\F_{q^m}}$. For $i\geq1$, the \emph{$i$-th Schur power} of $\mC$ is
    \[\mC^{(i)} := \mC\star\mC^{(i-1)}= \underbrace{\,\mC\star\cdots\star\mC\,}_{i \text{ times}}.\]
The Schur product between linear codes has been largely studied, for instance in~\cite{randriambololona2013asymptotically,randriambololona2013upper,Randriambololona, Mirandola-Zemor, Cascudo-Cramer-Mirandola-Zemor, couvreur2014distinguisher}. In particular, it is well-known that the dimension of the Schur square distinguishes algebraic structured linear codes, such as Reed-Solomon codes, from random ones (see~\cite[Corollary 27]{Mirandola-Zemor}). \\ In this context, we are particularly interested in the following definition.
\begin{definition}\label{definition:dimseq}
    Let $\mC\subseteq \F_{q^m}^n$ be a linear code. The sequence of integers $$\dim\big(\mC^{(i)}\big)\text{ for } i\geq0$$ is called the \emph{dimension sequence}, or the \emph{Hilbert sequence}, of $\mC$. The \emph{Castelnuovo-Mumford regularity} of $\mC$ is the smallest integer $r=r(\mC)\geq0$ such that, for every $t\geq r$, $$\dim\big(\mC^{(t)}\big)=\dim\big(\mC^{(r)}\big).$$
\end{definition}

\noindent Notice that the Castelnuovo-Mumford regularity is well-defined. It is easy to show that the dimension sequence of a nonzero linear code is nondecreasing, therefore it stabilizes after a finite number of steps~\cite[Proposition 11]{randriambololona2013asymptotically}. The terms \emph{Hilbert sequence} and \emph{Castelnuovo-Mumford regularity} are borrowed from commutative algebra, where analogous concepts are defined. \bigskip

\noindent Let $\PG(k-1,q)$ be the projective geometry over $\Fq^k$ defined as $$\PG(k-1,q):=\big(\Fq^k \setminus \{0\}\big) / _\sim ,$$ where $\sim$ denotes the proportionality relation such that, for $u,v\in\Fq^k$, $u \sim v$ if and only if $u = \lambda v$ for some $\lambda \in \Fq^*$. 
\begin{definition}
 Let $\mC\subseteq \F_{q^m}^n$ be a linear code with generator matrix $G=[\ g_1\mid g_2 \mid \cdots \mid g_n\ ]$, where each column $g_i$ is nonzero for all $i\in \{1,\ldots,n\}$. Let 
    \begin{equation*}
    \Pi_{G}:=\{\langle g_1^{\T}\rangle_{q^m},\dots,\langle g_n^{\T}\rangle_{q^m} \}\subseteq \PG(k-1,q^m),
\end{equation*}
    where $\langle g_i^{\T}\rangle_{q^m}$ is the projective point associated to $g_i^{\T}$. This is a {\em set of projective points} (without multiplicities) {\em associated to $\mC$ via $G$}. 
\end{definition}
We recall that the \emph{homogeneous coordinate ring} of a projective variety $V\subseteq \mathrm{PG}(k-1,q)$ is the quotient ring $\F_q[x_1,\dots,x_k]/\mI(V)$, where $\mI(V)$ is the homogeneous ideal of polynomials vanishing on $V$. The \emph{Hilbert function} of $V$ is given by $\mathrm{HF}_V(t):=\dim_{\F_q}(\F_q[x_1,\dots,x_k]/\mI(V))_t$, for $t\geq 0$. There exists a unique polynomial, called the Hilbert polynomial, which coincides with the Hilbert function for sufficiently large $t$. The smallest integer for which the Hilbert function agrees with the Hilbert polynomial is called \emph{index of regularity}. Since the index of regularity and the Castelnuovo-Mumford regularity coincide in our setting, we will adopt the index of regularity as the definition of Castelnuovo-Mumford regularity to avoid unnecessary complexity. For further details we refer to~\cite{eisenbud2005second}. The close connection between these objects is clarified in the following proposition.
\begin{proposition}[{\cite[Proposition 1.28]{Randriambololona}}]\label{proposition:dimseqalg}
    Let $\mC\subseteq \F_{q^m}^n$ be a linear code with generator matrix $G$ and let $\Pi_{G}\subseteq \PG(k-1,q^m)$ be defined as above. Then,
    \begin{enumerate}
        \item the dimension sequence of $\mC$ is equal to the Hilbert function of $\Pi_G$,
        \item the Castelnuovo-Mumford regularity $r(\mC)$ of $\mC$ is equal to the Castelnuovo-Mumford regularity of $\Pi_G$,
        \item $\dim\left(\mC^{(t)}\right)=\lvert \Pi_G\rvert$, for all $t\geq r(\mC)$,
    \end{enumerate}
    where we define the Hilbert function and the Castelnuovo-Mumford regularity of $\Pi_G$ as those of its homogeneous coordinate ring.
\end{proposition}

\begin{remark}\label{remark:dimseq}
    The dimension sequence is well-defined. Indeed, the Hilbert function of an ideal is invariant under projective automorphisms. This implies that the Hilbert function of the homogeneous coordinate ring of $\Pi_{G}$ does not depend on the choice of the generator matrix $G$. \\ Furthermore, the dimension of $\mC^{(i)}$ is invariant under monomial equivalence: if $\mC_1$ and $\mC_2$ are two linear codes with generator matrices $G_1$ and $G_2$, where $G_1=G_2M$ for a monomial matrix $M$, it is trivial to conclude that $\Pi_{G_1}=\Pi_{G_2}$. 
\end{remark}\medskip

\subsection{The associated Hamming-metric code}
In order to generalize the previous concept to rank-metric codes, we need to introduce the notion of Hamming-metric code associated to a rank-metric one, as well as the geometric interpretation of rank-metric codes. \\ Let $\mC$ be a nondegenerate $[n,k,d]_{q^m/q}$ rank-metric code. Let $G$ be a generator matrix of $\mC$ and $g_1,\dots,g_n\in\F_{q^m}^k$ be the columns of $G$. The $\F_q$-vector space $$\mU_G:=\langle g_1,\ldots,g_n\rangle_{\F_q}$$ has $\F_q$-dimension $n$ and $\langle \mU_G \rangle_{\F_{q^m}}=\F_{q^m}^k$. Thus, $\mU_G$ is naturally called an $[n,k]_{q^m/q}$-\emph{system} associated to the code $\mC$. Two $[n,k]_{q^m/q}$-systems $\mU$ and $\mU'$ are said to be \emph{equivalent} if there is an $\F_{q^m}$-isomorphism $\varphi:\Fm^k \rightarrow \Fm^k$ such that $\varphi(\mU)=\mU'$. Clearly, if $G$ and $G'$ are two generator matrices of the same code, $\mU_G$ and $\mU_{G'}$ are equivalent. As a consequence, with a little abuse of notation, we may drop the index and simply talk about the $[n,k]_{q^m/q}$-system $\mU$ associated to $\mC$. It is also straightforward to see that equivalent codes are associated to equivalent systems (for more details, see the appendix of~\cite{Alfarano-Borello}). \\
Let now $\sim_{\F_q}$ be the proportionality relation over $\F_{q^m}^k$ such that, for $u,v\in\F_{q^m}^k$, $u\sim_{\F_q} v$ if and only if $u=\lambda v$ for some $\lambda \in \Fq^*$.
\begin{definition}
    Let $\mU$ be the $[n,k]_{q^m/q}$-system associated to the nondegenerate  $[n,k,d]_{q^m/q}$ rank-metric code $\mC$. Let $G^{\rm H}(\mU)\in\F_{q^m}^{k\times N}$ be a matrix whose columns form a set of representatives for  the set of equivalence classes $\big(\mU\setminus \{0\}\big)/\sim_{\F_q}$. We will denote $G^{\rm H}(\mU)$ as an \emph{extended $($generator$)$ matrix} of $\mC   $. Moreover, let $\mC^{\rm H}$ be the $[N,k]_{q^m}$ code generated by $G^{\rm H}(\mU)$, where $N=(q^n-1)/(q-1)$. Then, we will denote $\mC^{\rm H}$ as a \emph{Hamming-metric code associated} to $\mC$. 
\end{definition}

\begin{proposition} \label{prop: uniqueness of associated code}
    Let $\mC$ be a nondegenerate $[n,k,d]_{q^m/q}$ rank-metric code. Then, its associated Hamming-metric code $\mC^{\rm H}$ is unique up to columns permutation and right multiplication by a diagonal matrix with entries in $\F_q^*$.
\end{proposition}
\begin{proof}
    Let $\mU$ be the $[n,k]_{q^m/q}$-system associated to $\mC$. Consider two matrices $G_1^{\rm H}(\mU),G_2^{\rm H}(\mU)\in\F_{q^m}^{k\times N}$ whose columns form two sets of representatives $\{g_{1,1},\ldots,g_{1,N}\}$ and $\{g_{2,1},\ldots,g_{2,N}\}$ for $\big(\mU\setminus \{0\}\big)/\sim_{\F_q}$. 
    By construction, these two sets are such that $g_{1,1} = \lambda_1 g_{2,{i_1}}, \dots, g_{1,N} = \lambda_N g_{2,{i_N}}$ for some indices permutation $\pi(1,\dots,N)=(i_1,\dots,i_N)$  and some $\lambda_1,\dots,\lambda_N\in\Fq^*$. Hence, $$G_1^{\rm H}(\mU)=
    \begin{pmatrix}
        g_{1,1} | \dots | g_{1,N}
    \end{pmatrix}=
    \begin{pmatrix}
        g_{2,{i_1}} | \dots | g_{2,{i_N}}
    \end{pmatrix}
     \begin{pmatrix}
        \lambda_1 & &\\
        & \ddots & \\
        & & \lambda_N
    \end{pmatrix}.
    $$
    In conclusion, the two Hamming-metric codes $\mC_1^{\rm H}= {\rm rowsp}\big(G_1^{\rm H}(\mU)\big)$ and $\mC_2^{\rm H}={\rm rowsp}\big(G_2^{\rm H}(\mU)\big)$ associated to $\mC$ satisfy the desired uniqueness condition.
\end{proof}
As a consequence, it makes sense again, with a little abuse of notation, to talk about the Hamming-metric code associated to $\mC$. For all the details about this object, see~\cite[Section 4.2.]{Alfarano-Borello}.

\medskip

\subsection{The geometry of rank-metric codes}
The associated Hamming-metric code is closely related to some geometric objects called \emph{linear sets}. These objects were introduced by Lunardon in~\cite{lunardon1999normal} to construct blocking sets and have been the subject of intense research in recent years. A detailed treatment of linear sets can be found in~\cite{polverino2010linear}.

\begin{definition}
Let $\mU$ be an $\Fmk$-system. The $\F_q$-\emph{linear set} of rank $n$ associated with $\mU$ is $$L_\mU:=\big\{\langle u^{\T} \rangle_{\F_{q^m}} \st  u\in \mU\setminus\{0\}\big\}\subseteq \PG(k-1, q^m).$$ 
Two linear sets are said to be \emph{equivalent} if their systems are equivalent.
\end{definition}

\begin{remark} 
The original definition of a linear set does not require that $\langle \mU\rangle_{\Fm} = \Fm^k$. If $\dim_{\Fm}(\langle \mU\rangle_{\Fm}) = h < k$, we can assume, without loss of generality, that $\mU \subseteq \F_{q^m}^h$. In this case, we can then study $L_\mU$ within $\PG(h-1,q^m)$. 
\end{remark}

\begin{definition}
    Let $W$ be an $\Fm$-linear subspace of $\Fm^k$ and $\Lambda = \PG(W, \Fm)$ be the pro\-jec\-tive sub\-space corresponding to $W$. The {\em weight of $\Lambda$ in $L_\mU$} is defined as \[{\rm w}_\mU(\Lambda) := \dim_{\Fq}(\mU\cap W).\]
\end{definition}
\noindent It is immediate to see that the linear set associated with any $\Fmk$-system $\mU$ satisfies \[|\,L_\mU\,|\leq\frac{q^n-1}{q-1}.\] A linear set $L_\mU$ matching with equality the previous bound is said to be {\em scattered}. Equivalently, $L_\mU$ is scattered if and only if ${\rm w}_\mU(P) = 1$ for all $P\in L_\mU$. Scattered linear sets can be further generalized via the notion introduced in \cite{sheekey2020rank} of {\em scattered with respect to $h$-dimensional subspaces} or {\em $h$-scattered}. In particular, a linear set $L_\mU\subseteq\PG(k-1,q^m)$ is said to be {\em scattered with respect to hyperplanes} if ${\rm w}_\mU(\mathcal{H}) \leq k-1$, for every hyperplane $\mathcal{H}$ of $\PG(k-1,q^m)$.

\medskip

\noindent The linear set $L_\mU$ is the \emph{geometric object associated} to the rank-metric code $\mC$ whose generator matrix is $G$. Notice that changing the generator matrix trivially results in equivalent linear sets. Moreover, equivalent rank-metric codes correspond to equivalent associated linear sets. The rank metric can be deduced from linear sets by examining their intersections with hyperplanes (taking into account weights of points, see for example~\cite{Alfarano-Borello}). In our context, the following remark is particularly important.

\begin{remark}\label{remark:linearset}
Let $\mC$ be a $[n,k]_{q^m/q}$ code, $G$ be a generator matrix of $\mC$ and $\mC^{\rm H}$ be the Hamming-metric code associated to $\mC$, with generator matrix $G^{\rm H}$. Then, 
$$L_{\mU_G}=\Pi_{G^{\rm H}}.$$ 
\end{remark}

\begin{example}\label{ex:basic}
    Let $q=2, n=m=3$ and $k=2$. Consider $\F_8=\F_2(\alpha)$, where $\alpha^3+\alpha+1=0$. Let $\mC$ be the $[3,2]_{8/2}$ code generated by \[G=\begin{pmatrix}
        1 & 0 & \alpha \\ 0 & 1 & 1
    \end{pmatrix}.\] Then, the $[3,2]_{8/2}$-system spanned by the columns of $G$ is \[{\mU_G}=\left\{\begin{pmatrix}0 \\ 0\end{pmatrix},\begin{pmatrix}1 \\ 0\end{pmatrix}, \begin{pmatrix}0 \\ 1\end{pmatrix}, \begin{pmatrix}\alpha \\ 1\end{pmatrix}, \begin{pmatrix}1 \\ 1\end{pmatrix}, \begin{pmatrix}\alpha+1 \\ 1\end{pmatrix}, \begin{pmatrix}\alpha \\ 0\end{pmatrix}, \begin{pmatrix}\alpha+1 \\ 0\end{pmatrix}\right\}\]
    and the set of its equivalence classes with respect to $\sim_{\F_2}$ is given by
    \[({\mU_G}\setminus\{0\})/\sim_{\F_2}=\left\{\begin{bmatrix}1 \\ 0\end{bmatrix}, \begin{bmatrix}0 \\ 1\end{bmatrix}, \begin{bmatrix}\alpha \\ 1\end{bmatrix}, \begin{bmatrix}1 \\ 1\end{bmatrix}, \begin{bmatrix}\alpha+1 \\ 1\end{bmatrix}, \begin{bmatrix}\alpha \\ 0\end{bmatrix}, \begin{bmatrix}\alpha+1 \\ 0\end{bmatrix}\right\}.\]
    
    \noindent Hence, the extended generator matrix of $\mC$ is \[G^{\rm H}=\begin{pmatrix}
        1 & 0 & \alpha & 1 & \alpha+1 & \alpha & \alpha+1 \\ 0 & 1 & 1 & 1 & 1 & 0 & 0
    \end{pmatrix},\]
    while its associated geometric object is the linear set \[L_{\mU_G}=\left\{ (1:0), (0:1),(1:1), (\alpha:1), (\alpha+1:1)\right\}\subseteq\PG(1, 8).\] It can then be immediately seen that $$L_{\mU_G}=\Pi_{G^{\rm H}},$$ whose vanishing ideal is \[\mI(L_{\mU_G})=\big(x_1(x_1+x_2)(x_1+\alpha x_2)(x_1+(\alpha+1)x_2)x_2\big)\subseteq\F_8[x_1,x_2].\] \\ Note also that in this example $|\,L_{\mU_G}\,|=5<7$ and we can see that \[{\rm w}_\mU\big((1:0)\big)=\dim_{\F_2}\big(\mU\cap\langle(1,0)\rangle_{\F_8}\big)=2.\] Hence,  $L_{\mU_G}$ is not a scattered linear set.
\end{example}

In the next section we will explore some geometric invariants of these objects. \bigskip

\section{$\F_q$-dimension sequence of rank-metric codes} \label{section:dimension sequence of rank-metric codes}
As we mentioned earlier, the dimension sequence is an important geometric invariant of Hamming-metric codes, which may be used to differentiate between algebraic structured codes and random ones. It might seem natural to define the dimension sequence similarly for rank-metric codes. However, this approach presents two relevant drawbacks:
\begin{itemize}
    \item the dimension sequence often converges too quickly, so there is not enough ``space" to discriminate between families of MRD codes and random ones;
    \item even though the dimension sequence is invariant under monomial equivalences, as observed in Remark~\ref{remark:dimseq}, it is not invariant under rank-metric equivalences, as highlighted in the following example.
\end{itemize}

\begin{example}
    Let $\mC_1=\langle(1,0,1),(1,1,0)\rangle_{\F_{q^m}}$ and $\mC_2=\langle(1,0,0),(0,1,0)\rangle_{\F_{q^m}}$ be two $[3,2,1]_{q^m/q}$ codes. It is easy to see that, although the two codes are rank-metric equivalent, we have $$\dim\left(\mC_1^{(2)}\right)=3 \ \text{ and }\ \dim\left(\mC_2^{(2)}\right)=2.$$
\end{example}

\noindent To fix these problems, in this section we propose the following definition of $\F_q$-dimension sequence for rank-metric codes.
\begin{definition}
    The $\F_q$-\emph{dimension sequence}, or $\F_q$-\emph{Hilbert sequence}, $\{h_i(\mC)\}_{i\geq0}$ of a nondegenerate rank-metric code $\mC$ over $\F_{q^m}$ is the Hilbert sequence of the associate Hamming-metric code $\mC^{\rm H}$, that is
    $$h_i(\mC):=\dim_{\F_{q^m}}\left(\mC^{{\rm H}(i)}\right).$$
    Moreover, the $\F_q$-\emph{Castelnuovo-Mumford regularity} $r(\mC)$ of $\mC$ is the Castelnuovo-Mumford regularity of $\mC^{\rm H}$.
\end{definition}

\begin{remark}
    Notice that the $\F_q$-dimension sequence and the $\F_q$-Castelnuovo-Mumford regularity do not depend on the choice of $\mC^{\rm H}$, hence they are well-defined. Consider two equivalent codes $\mC_1$ and $\mC_2$ and let $G_1^{\rm H}$ and $G_2^{\rm H}$ be two extended matrices of $\mC_1$ and $\mC_2$ respectively. We have already noticed in Proposition \ref{prop: uniqueness of associated code} that extended codes are unique up to monomial equivalences. Then, the columns of $G_1^{\rm H}$ must be equal to those of $G_2^{\rm H}$ up to permutation and right scalar multiplication. This implies that $\Pi_{G_1^{\rm H}}=\Pi_{G_2^{\rm H}}$, and therefore $h_i(\mC_1)=h_i(\mC_2)$, for all $i\geq 0$. Looking at Remark~\ref{remark:linearset}, we realize that we are merely considering the dimension sequence of the linear set $L_{\mU_G}$ associated with the code $\mC$ via a generator matrix $G$.
\end{remark}

\begin{theorem} \label{theorem:hi(C)}
    Let $\mathcal{C}\subseteq\F_{q^m}^n$  be a code of dimension $k$ and generator matrix $G$. Let also $\mathcal{I}(L_{\mU_G})$ be the vanishing ideal of $L_{\mU_G}$ in $\F_{q^m}[x_1,\dots,x_k]$. Then, for every positive integer $i$, $$h_{i}(\mC)=\binom{k+i-1}{i}-\dim_{\F_{q^m}}\big(\mathcal{I}(L_{\mU_G})_{i}\big),$$ where $\mathcal{I}(L_{\mU_G})_{i}$ is the set of all homogeneous polynomials in $\mathcal{I}(L_{\mU_G})$ of degree $i$. In particular, for $i\in\{0,\dots,q\}$, $$h_{i}(\mC)=\binom{k+i-1}{i}.$$
\end{theorem}
\begin{proof}
    By Proposition~\ref{proposition:dimseqalg} and Remark~\ref{remark:linearset}, we know that the $\F_q$-dimension sequence $\{h_i(\mC)\}_{i\geq0}$ is equal to the dimension sequence of $L_{\mU_G}$. Hence,
    $$h_{i}(\mC)=\dim_{\F_{q^m}}\F_{q^m}[x_1,\dots,x_k]_i - \dim_{\F_{q^m}}\big(\mathcal{I}(L_{\mU_G})_{i}\big),$$
    where $\F_{q^m}[x_1,\dots,x_k]_i$ is the set of homogeneous polynomials of degree $i$, so that
    $$h_{i}(\mC)=\binom{k+i-1}{i}-\dim_{\F_{q^m}}\big(\mathcal{I}(L_{\mU_G})_{i}\big).$$
    
    \noindent As we remarked above, the dimension sequence does not depend on the choice of $G$. Hence, we can assume that $\PG(k-1,q)\subseteq L_{\mU_G}$ (it is sufficient to consider a generator matrix $G$ in systematic form). This implies that $$\mathcal{I}(L_{\mU_G})\subseteq \mathcal{I}(\PG(k-1,q))\subseteq \F_{q^m}[x_1,\dots,x_k].$$ Therefore, knowing the ideal $\mathcal{I}(\PG(k-1,q))$ implies knowing what kind of polynomials we can expect to find in $\mathcal{I}(L_{\mU_G})$. Let $$\mathcal{F}_s:=\big\langle x_i^{[s]}x_j-x_ix_j^{[s]}\st 1\leq i<j\leq k\big\rangle_{\F_{q^m}}.$$ By~\cite[Theorems 2.5 and 2.8]{beelen2019vanishing}, we know that the vanishing ideal over $\F_{q^m}$ of $\PG(k-1,q)$ is the ideal of $\F_{q^m}[x_1,\dots,x_k]$ generated by $\mathcal{F}_1$. Hence, if we restrict ourselves to the case $i\in\{0,\dots,q\}$, we find that $\mathcal{I}(L_{\mU_G})_{i}=\{0\}$, hence the second part of the theorem holds.
\end{proof}

As a consequence, the first significant element in the $\F_q$-dimension sequence is the $(q+1)$-th. As we will see later, this is also sufficient to discriminate between Gabidulin and random codes. 

\begin{example}
We consider again the code of Example \ref{ex:basic}, that is, the $[3,2]_{8/2}$ code $\mC$ generated by \[G=\begin{pmatrix}
        1 & 0 & \alpha \\ 0 & 1 & 1
    \end{pmatrix}.\] 
Recall that the vanishing ideal of $L_{\mU_G}$ is \[\mI(L_{\mU_G})=\big(x_1(x_1+x_2)(x_1+\alpha x_2)(x_1+(\alpha+1)x_2)x_2\big)\subseteq\F_8[x_1,x_2].\] In this case, the corresponding $\F_2$-Hilbert sequence is 
$$\{1,2,3,4,5,5,5,5,\ldots\},$$
while the sequence $\big\{\dim_{\F_{8}}\big(\mathcal{I}(L_{\mU_G})_{i}\big):i\geq 0\big\}$ is 
$$\{0,0,0,0,0,1,2,3,\ldots\}.$$
This is obtained by observing that the ideal is principal and generated by a polynomial of degree $5$.
\end{example}

The next theorem provides a simple way to compute $h_{q+1}(\mC)$. We will postpone its proof to the next section, where we will study the link between Overbeck's invariant (Theorem \ref{theorem:randomqsum}), $h_{q+1}(\mC)$, and the $\mathcal{F}_s$ spaces.
\begin{theorem}\label{theorem:codeq+1}
    Let $\mathcal{C}\subseteq\F_{q^m}^n$  be a code of dimension $k$ with generator matrix $[\,I_k\,|\, X\,]$, where $X\in\F_{q^m}^{k\times (n-k)}$, and $r=\mathrm{rk}(X^{[1]}-X)$. Then,
    $$
h_{q+1}(\mC)=\binom{k+q}{q+1}-\binom{k-r}{2}.$$
\end{theorem}
\begin{proof}
    See next section.
\end{proof}

\noindent The previous theorem allows us to compute the exact value of $h_{q+1}(\mG_k)$ for a Gabidulin code $\mG_k$.

\begin{corollary}
    Let $\mathcal{C}\subseteq\F_{q^m}^n$  be an MRD code of dimension $k<n$. Then, $\mC$ is a Gabidulin code if and only if $$h_{q+1}(\mC)=\binom{k+q}{q+1}-\binom{k-1}{2}.$$
\end{corollary}
\begin{proof}
    Without loss of generality, we can assume that $\mC$ has a generator matrix in systematic form, that is $[\,I_k\,|\, X\,]$. Since $\mC$ is MRD, by Lemma~\ref{lm:rank1}, we know that $\mathrm{rk}(X^{[1]}-X)=1$ if and only if $\mC$ is a Gabidulin code. Therefore, we immediately get the result via Theorem~\ref{theorem:codeq+1}.
\end{proof}
\bigskip

Now that the Gabidulin case has been fully studied up to the $(q+1)$-th term of the $\Fq$-dimension sequence, it is natural to investigate the general behavior of a random code. To address this question, we rely on the analysis made in~\cite{neri2018genericity}, where the following result is specifically presented in the same scenario that fits our framework.
\begin{lemma}[Schwartz-Zippel, Corollary 1 in \cite{schwartz1980fast}] \label{lemma:Schwartz-Zippel}
    Let $p\in\F_q[x_1,\dots,x_k]$ be a nonzero polynomial of degree $d\geq0$. Let $\F_{q^m}$ be an extension field and $S\subseteq\F_{q^m}$ be a finite set. If $s_1,\dots,s_k$ are selected independently and uniformly at random from $S$, then $$\Pr\big(p(s_1,\dots,s_k)=0\big)\leq\frac{d}{|S|}.$$
\end{lemma}

\noindent Hence, we obtain the following.
\begin{lemma} \label{lemma:probability random}
    Let $\mathcal{C}\subseteq\F_{q^m}^n$ be a random code of dimension $k$, with generator matrix $[\, I_k \, | \, X\, ]$. Let also $r=\min\{k,n-k\}$ and $s\in\{1,\dots,m-1\}$ with $\gcd(s, m) = 1$. Then, $$\Pr\big(\rk(X^{[s]}-X)=r\big)\ge 1-\frac{r}{q^{m-1}}.$$
\end{lemma}
\begin{proof}
    Since $\mC$ is random, the entries of $X$ have been chosen independently and uniformly at random from $\Fm$. When considering the matrix $X^{[s]}-X$ (instead of $X$), this translates into having selected independently and uniformly at random elements in $S=\{\alpha\in\Fm \st \mathrm{Tr}_{q^m/q}(\alpha)=0\}$. Indeed, it is well-known that $|\,S\,|=q^{m-1}$, and the linear map $\phi:x\mapsto x^{[s]}-x$ from $\Fm$ to $S$ is surjective with $\mathrm{ker}(\phi)=\Fq$. This directly implies that for any $\alpha \in S$, the preimage $\phi^{-1}(\alpha)$ has cardinality $q$, see~\cite[Section 2.3]{lidl1994introduction} for further details. Since any element in the preimage of $\alpha$ can be chosen with probability $1/q^m$, we get that the probability for a given entry of $X^{[s]}-X$ to be equal to $\alpha$ is $q\cdot 1/q^m=1/q^{m-1}$. \\ Let now $r=\min\{k,n-k\}$ and consider any $r\times r$ block $B$ of $X^{[s]}-X$. Then, $$\Pr\big(\rk(X^{[s]}-X)=r\big)\ge \Pr\big(\det(B)\ne 0\big)=1-\Pr\big(\det(B)=0\big).$$ If we regard the entries of $B$ as the variables $x_1,\dots,x_{r^2}$, we can look at $\det(B)$ as a polynomial in $\Fq[x_1,\dots,x_{r^2}]$ of degree $r$. Then, by Schwartz-Zippel's lemma, we get that $\Pr\big(\det(B)=0\big)\le\frac{r}{|\,S\,|}$, which directly implies the theorem.
\end{proof}

\begin{corollary}
    Let $\mathcal{C}\subseteq\F_{q^m}^n$ be a random code of dimension $k$ and $r=\min\{k,n-k\}$. Then, $$\Pr\bigg(h_{q+1}(\mC)=\binom{k+q}{q+1}-\binom{k-r}{2}\bigg)\ge1-\frac{r}{q^{m-1}}.$$
\end{corollary}
\begin{proof}
    Given that $[\, I_k \, | \, X\, ]$ is a generator matrix of $\mC$, by Theorem \ref{theorem:codeq+1} it follows that $$\Pr\big(\rk(X^{[1]}-X)=r\big)=\Pr\bigg(h_{q+1}(\mC)=\binom{k+q}{q+1}-\binom{k-r}{2}\bigg),$$ where $\binom{k-r}{2}$ is meant to be equal to $0$ when $r\ge k-1$. Hence, the result is a direct consequence of Lemma \ref{lemma:probability random}.
\end{proof}

\begin{remark}
    In this paper we focus on linear rank-metric codes. However, the dimension sequence can be also used to differentiate between a linear set and a random set of points. In particular, we expect the Castelnuovo-Mumford regularity of a random set of points to be strictly lower than that of a linear set.
\end{remark}
\bigskip

In order to support our theoretical findings, we conducted a series of tests to study and compare the $\Fq$-dimension sequences of Gabidulin and random codes. In this analysis, we extended our study beyond the $(q+1)$-th Schur power of codes, exploring higher powers as well. These investigations were carried out using {\sc Magma}~\cite{bosma1997magma}, see~\cite{CiteGithub} for more details about the implementation, and the results are presented in Figure~\ref{figure}. \\ In the figure, the orange points represent the terms of the $\Fq$-Hilbert sequence associated with Gabidulin codes, while the blue dots are related with random codes. Specifically, for $q=2$, our experiments showed that in $98.8\%$ of the $10000$ tests conducted, the resulting $\Fq$-Hilbert sequence exactly matched the one shown in blue in the graph. For $q=3$, the same procedure always yielded the same result, hence the sequence displayed in the graph corresponds to $100\%$ of the outcomes observed in our experiments. This suggests that the sequences shown in the figure are representative of the typical behavior of random codes in their respective settings. \\ These experimental results confirm our theoretical observations regarding the first $(q+1)$-th terms of the $\mathbb{F}_q$-dimension sequences for both Gabidulin and random codes. Furthermore, the figure shows that the two $\mathbb{F}_q$-dimension sequences differ in their $\Fq$-Castelnuovo-Mumford regularities, with random codes exhibiting a strictly smaller value than Gabidulin codes.

\begin{figure}[ht]
    \centering
    \includegraphics[width = 16 cm]{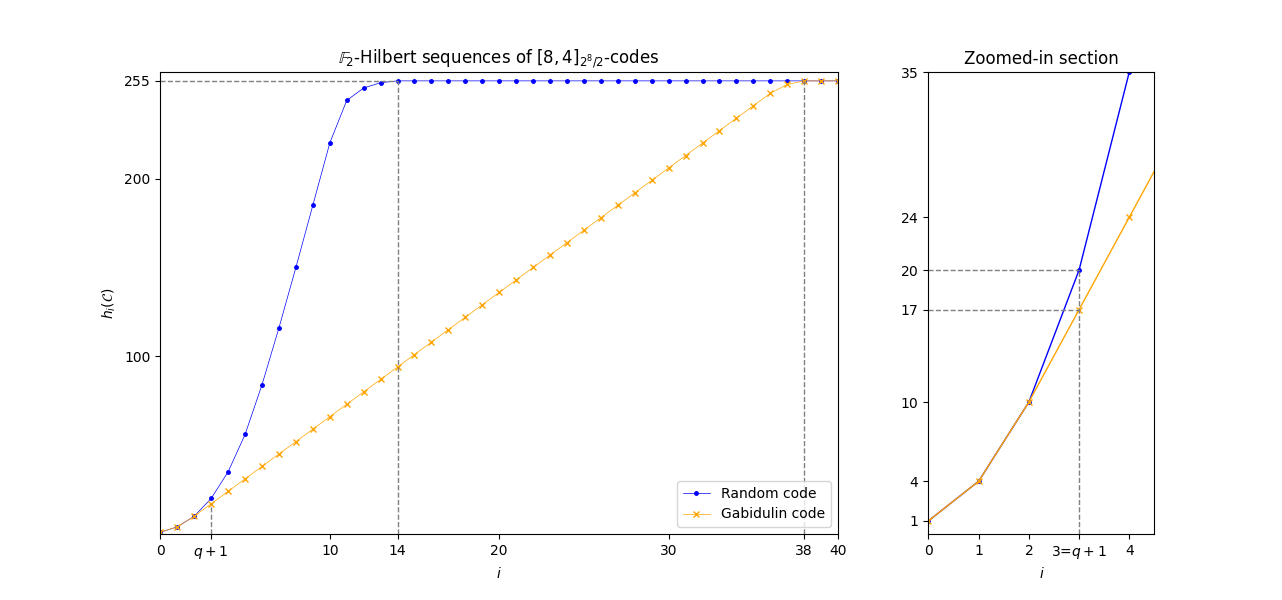}
    \includegraphics[width = 16 cm]{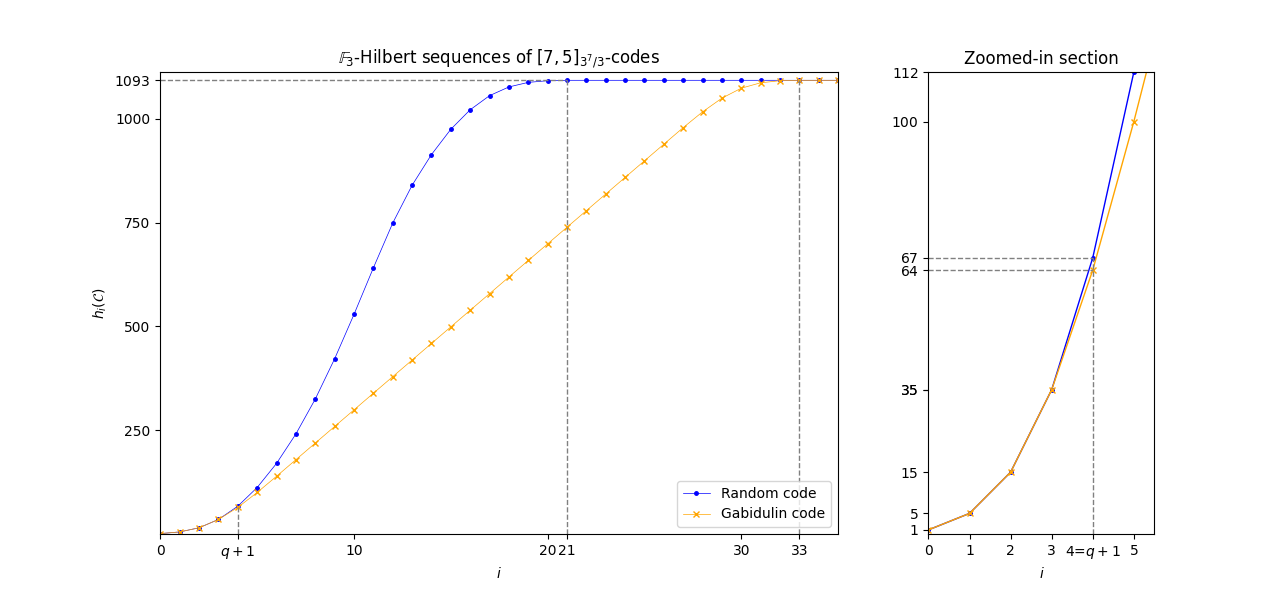}
    \caption{Experimental investigations.} \label{figure}
\end{figure}

\begin{remark}
    In this section, we have shown that the $(q+1)$-th term of the $\mathbb{F}_q$-dimension sequence characterizes Gabidulin codes. A natural question that follows is whether the $\mathbb{F}_q$-Castelnuovo-Mumford regularity could serve a similar purpose. So far, the Castelnuovo-Mumford regularity of codes has been studied by Randriambololona in the context of Reed-Solomon codes~\cite{Randriambololona}. The Gabidulin case is more complex, and obtaining the desired results in this context requires further investigations.
\end{remark}
\bigskip

\section{Evaluation of $(q^s+1)$-forms on linear sets} \label{section:Evaluation of forms on linear sets}

As we have seen in Theorem~\ref{theorem:randomqsum}, the dimension of the first $q$-sum of a given code allows to decide whether it is a Gabidulin or a random one. In particular, to compute the dimension of the first $q$-sum of a Gabidulin code $\mG_k$ with generator matrix $[\,I_k\,|\, X\,]$, it is sufficient to notice that
$$\dim(\mG_k+\mG_k^{[1]})=k+\rk(X^{[1]}-X),$$
and then apply Lemma~\ref{lm:rank1}. This gives   $\dim(\mG_k+\mG_k^{[1]})=k+1$, whereas the dimension of the first $q$-sum is $2k$ in most other cases. A careful reader could have already noticed that the condition $\rk(X^{[1]}-X)=1$ appearing in Lemma~\ref{lm:rank1} is the same condition that is stated in Theorem~\ref{theorem:codeq+1}. In this section, we will show that $\rk(X^{[s]}-X)$ is related to the dimension of the subspaces of $\mathcal{F}_s$ (where we recall that $\mathcal{F}_s:=\langle x_i^{[s]}x_j-x_ix_j^{[s]}\st 1\leq i<j\leq k\rangle_{\F_{q^m}}$) vanishing on the points of the linear set associated to the code generated by $[\,I_k\,|\, X\,]$. In the particular case of $s=1$, this will result in a proof of Theorem~\ref{theorem:codeq+1}. \medskip

\noindent We begin by studying the set of forms in $\mathcal{F}_s$ vanishing on a coset of $\F_q^k$. In the following, we align to the convention that an empty sum is equal to zero.

\begin{lemma}\label{lm:aU}
A $(q^s+1)$-form $$p(x):=\sum_{1\leq i<j\leq k}A_{i,j}(x_i^{[s]}x_j-x_ix_j^{[s]})\in\mathcal{F}_s$$ vanishes on $\mathcal{V}_{\alpha}:=\alpha+\F_q^k$, with $\alpha=(\alpha_1,\dots,\alpha_k)\in\F_{q^m}^k$ if and only if 
\begin{equation*}
\begin{dcases}
\sum_{i=1}^{t-1}A_{i,t}(\alpha_i^{[s]}-\alpha_i)-\sum_{j=t+1}^kA_{t,j}(\alpha_j^{[s]}-\alpha_j)=0&t\in \{1,\ldots,k\}
\end{dcases}.    
\end{equation*}
\end{lemma}
\begin{proof}
Consider a $(q^s+1)$-form 
$$p(x):=\sum_{1\leq i<j\leq k}A_{i,j}(x_i^{[s]}x_j-x_ix_j^{[s]})\in\mathcal{F}_s$$
vanishing on $\mV_\alpha = \alpha+\F_q^k$. Note that, in particular, this implies that $p(x)$ vanishes on $\alpha$. For $a\in\F_q^k$, by evaluating $p(x)$ at $\alpha+a$, we obtain
$$0=p(\alpha+a)=p(\alpha)+\sum_{1\leq i<j\leq k}A_{i,j}((\alpha_i^{[s]}-\alpha_i)a_j-(\alpha_j^{[s]}-\alpha_j)a_i).$$
Since $p(\alpha)=0$, we have
$$\sum_{1\leq i<j\leq k}\left(A_{i,j}(\alpha_i^{[s]}-\alpha_i)a_j-A_{i,j}(\alpha_j^{[s]}-\alpha_j)a_i\right)=0,$$
implying that
$$\sum_{t=1}^{k}\left(\sum_{i=1}^{t-1}A_{i,t}(\alpha_i^{[s]}-\alpha_i)-\sum_{j=t+1}^kA_{t,j}(\alpha_j^{[s]}-\alpha_j)\right)a_t=0.$$
Since this is true for every $a\in\F_q^k$ and  $p(\alpha)=0$, we get the following linear system
\begin{equation*}
\begin{dcases}
\sum_{i=1}^{t-1}A_{i,t}(\alpha_i^{[s]}-\alpha_i)-\sum_{j=t+1}^kA_{t,j}(\alpha_j^{[s]}-\alpha_j)=0&t\in \{1,\ldots,k\}\\
\sum_{1\leq i<j\leq k}A_{i,j}(\alpha_i^{[s]}\alpha_j-\alpha_i\alpha_j^{[s]})=0&
\end{dcases}.    
\end{equation*}
To conclude, it is sufficient to observe that the last equation is a linear combination of the first $k$ ones. Indeed, one can easily check that
\begin{equation*}
\sum_{1\leq i<j\leq k}A_{i,j}(\alpha_i^{[s]}\alpha_j-\alpha_i\alpha_j^{[s]})=\sum_{t=1}^{k}\left(\sum_{i=1}^{t-1}A_{i,t}(\alpha_i^{[s]}-\alpha_i)-\sum_{j=t+1}^kA_{t,j}(\alpha_j^{[s]}-\alpha_j)\right)\alpha_t.
\end{equation*}
To prove the converse it is sufficient to compute $p(\alpha+a)$ and to check that it is a linear combination of the equations appearing in the system.
\end{proof}

\noindent The next lemma shows that if a form in $\mathcal{F}_s$ vanishes at two cosets $\mV_{\alpha_1}$ and $\mV_{\alpha_2}$, then it vanishes at the coset of any linear combination of $\alpha_1$ and $\alpha_2$ over $\F_q$. Equivalently, the set of cosets at which a form in $\mathcal{F}_s$ vanishes is an $\F_q$-vector space.
\begin{lemma}\label{lm:linearity}
Let $\alpha_1=(\alpha_{1,1},\dots,\alpha_{1,k})$ and $\alpha_2=(\alpha_{2,1},\dots,\alpha_{2,k})$ be vectors in $\F_{q^m}^k$.
A $(q^s+1)$-form $p(x)\in\mathcal{F}_s$ vanishes on $\mV_{\alpha_1}$ and on $\mV_{\alpha_2}$ if and only if it vanishes on $\mV_{\lambda_1\alpha_1+\lambda_2\alpha_2}$ for any $\lambda_1,\lambda_2\in\F_q$. Equivalently, a form $p(x):=\sum_{1\leq i<j\leq k}A_{i,j}(x_i^{[s]}x_j-x_ix_j^{[s]})\in\mathcal{F}_s$ vanishes on $\mV_{\lambda_1\alpha_1+\lambda_2\alpha_2}$ for any $\lambda_1,\lambda_2\in\F_q$, if and only if it satisfies the following system
\begin{equation*}
\begin{dcases}
\sum_{i=1}^{t-1}A_{i,t}(\alpha_{1,i}^{[s]}-\alpha_{1,i})-\sum_{j=t+1}^kA_{t,j}(\alpha_{1,j}^{[s]}-\alpha_{1,j})=0&t\in \{1,\ldots,k\}\\
\sum_{i=1}^{t-1}A_{i,t}(\alpha_{2,i}^{[s]}-\alpha_{2,i})-\sum_{j=t+1}^kA_{t,j}(\alpha_{2,j}^{[s]}-\alpha_{2,j})=0&t\in \{1,\ldots,k\}
\end{dcases}.
\end{equation*}
\end{lemma}
\begin{proof}
First of all, we observe that if $p(x)$ vanishes on $\mV_\alpha$, then for any $\lambda\in\F_q$, it also vanishes on $\mV_{\lambda\alpha}$. Indeed, let  $\lambda\in\F_q^*$ ($\lambda=0$ is obvious), since $p(x)$ is homogeneous we have $p(\lambda\alpha+a)=\lambda^2 p(\alpha+a/\lambda)$, for any $a\in\F_q^k$. \\ Assume that $p(x)$ vanishes on $\mV_{\alpha_1}$ and $\mV_{\alpha_2}$. Then, by Lemma~\ref{lm:aU}, we have
\begin{equation*}
\begin{dcases}
\sum_{i=1}^{t-1}A_{i,t}\bar\alpha_{1,i}-\sum_{j=t+1}^kA_{t,j}\bar\alpha_{1,j}=0&t\in \{1,\ldots,k\}\\
\sum_{i=1}^{t-1}A_{i,t}\bar\alpha_{2,i}-\sum_{j=t+1}^kA_{t,j}\bar\alpha_{2,j}=0&t\in \{1,\ldots,k\}
\end{dcases},
\end{equation*}
where $\bar y:=y^{[s]}-y$. Since $\overline{\lambda_1\alpha_{1,i}+\lambda_2\alpha_{2,i}}=\lambda_1\bar\alpha_{1,i}+\lambda_2\bar\alpha_{2,i}$ for any $\lambda_1,\lambda_2\in\F_q$, by appropriately adding the previous equations, we obtain
\begin{equation*}
\begin{dcases}
\sum_{i=1}^{t-1}A_{i,t}\overline{\lambda_1\alpha_{1,i}+\lambda_2\alpha_{2,i}}-\sum_{j=t+1}^kA_{t,j}\overline{\lambda_1\alpha_{1,j}+\lambda_2\alpha_{2,j}}=0&t\in \{1,\ldots,k\}
\end{dcases}.
\end{equation*}
By Lemma~\ref{lm:aU}, we conclude that if  $p(x)$  vanishes on $\mV_{\alpha_1}$ and $\mV_{\alpha_2}$, then it vanishes on $\mV_{\lambda_1\alpha_1+\lambda_2\alpha_2}$ for any $\lambda_1,\lambda_2\in\F_q$.
\end{proof}

\noindent We now proceed to prove the main result of this section.
\begin{theorem}\label{th:MAIN}
Let $\mathcal{C}\subseteq\F_{q^m}^n$  be a code of dimension $k$ with generator matrix $G=[\,I_k\,|\, X\,]$, where $X\in  \F_{q^m}^{k\times (n-k)}$, and let $s\in\{1,\dots,m-1\}$. 
Let also $\mathcal{I}(L_{\mU_G})$ be the vanishing ideal of the linear set $L_{\mU_G}$ (which contains $\PG(k-1,q)$). Then, for $r\in\{0,\dots,k\}$,
$$
\mathrm{rk}(X^{[s]}-X) = r \ \text{ if and only if } \ \dim \mathcal{F}_s\cap \mathcal{I}(L_{\mU_G})=\binom{k-r}{2},$$
where $\binom{k-r}{2}=0$ for $r=k-1$ or $r=k$.
\end{theorem}
\begin{proof}
Suppose that $\mathrm{rk}(X^{[s]}-X) = r$ and let $\alpha_i=(\alpha_{i,1},\dots,\alpha_{i,k})^T$ be the $i$-th column of $X$.  Up to permuting the columns of $X$, we can assume that $\bar \alpha_1,\dots,\bar\alpha_r$, where $\bar y:=y^{[s]}-y$, are linearly independent. By Lemma \ref{lm:linearity}, a polynomial $p(x)=\sum_{1\leq i<j\leq k}A_{i,j}(x_i^{[s]}x_j-x_ix_j^{[s]})\in\mathcal{F}_s$
vanishes on $L_{\mU_G}$ if and only if the following system is satisfied
\begin{equation}\label{eq:systemwhole}
\begin{dcases}
\sum_{i=1}^{t-1}A_{i,t}\bar\alpha_{\ell,i}-\sum_{j=t+1}^kA_{t,j}\bar\alpha_{\ell,j}=0, &\ell \in \{1,\dots,r\},\,t\in \{1,\ldots,k\}.
\end{dcases}
\end{equation}
The desired conclusion follows by determining the dimension of the solution space of the previous linear system. Since the solution space is invariant under row operations, up to a permutation of the indices, we can suppose that $\bar\alpha_1=(\bar\alpha_{1,1},\dots,
\ba_{1,k})^T$, $\ba_2=(0,\ba_{2,2},\dots,\ba_{2,k})^T,\dots,\ba_r=(0,\dots,0,\ba_{r,r},\dots,\ba_{r,k})^T$, where $\ba_{i,i}\ne 0$  $i\in\{1,\dots,r\}$. In this way, we obtain $r$ subsystems of \eqref{eq:systemwhole} with associated matrices
$$
M_1=\big[M_{1,1}|M_{1,2}|\cdots|M_{1,k-1}\big],\dots,M_r=\big[M_{r,1}|M_{r,2}|\cdots|M_{r,k-1}\big],
$$
where $M_{t,l}$ is a $k\times (k-l)$ matrix such that
$$ M_{t,l}=
\begin{pmatrix}
    0&&&\cdots&&0\\
    \vdots&&&&&\vdots\\
    0&&&\cdots&&0\\
    0&\cdots&0&-\ba_{t,t}&\cdots&-\ba_{t,k}\\
    0&&&\ddots&&\\
    \vdots&&&&\ddots&\\
    0&&&\cdots&&0
\end{pmatrix} \ \text{ or } \ M_{t,l}=
\begin{pmatrix}
    0&\cdots&0\\
    \vdots&&\vdots\\
    0&\cdots&0\\
    -\ba_{t,l+1}&\cdots&-\ba_{t,k}\\
    \ba_{t,l}&&\\
    &\ddots&\\
    &&\ba_{t,l}
\end{pmatrix}, $$ 
for $l< t$ and $l\ge t$ respectively. \\ Let us now consider the matrix
$$ M =
\begin{pmatrix}
M_1\\
\vdots\\
M_r\end{pmatrix}
$$
associated to the whole system \eqref{eq:systemwhole}. The rank of $M$ is at least $k-1+k-2+\dots+k-r$ since the pivot elements $\ba_{i,i}$ are different from zero in the last $k-i$ rows of the block matrices $M_{i,i}$. We now prove that the rank of $M$ is at most $k-1+k-2+\dots+k-r$. First of all, note that in the matrix $M_i$ the $l$-th row, for $l\in\{1,\dots,i\}$, depends on the last $k-i+1$ rows of $M_l$ and on the last $k-l$ rows of $M_i$. In particular, the left kernel of $M$ contains the vectors of the form
$$
(0,\dots,0|\cdots|\overbrace{\underbrace{0,\dots,0}_{i-1},\ba_{i,i},\dots,\ba_{i,k}}^{l\text{-th block}}|0,\dots,0|\overbrace{ \underbrace{0,\dots,0}_{l-1},\ba_{l,l},\dots,\ba_{l,k}}^{i\text{-th block}}|\cdots|0,\dots,0)\in\F_{q^m}^{rk},
$$
for $1\le l\le i\le r$, that is
$$ \begin{aligned}
    &(\ba_{1,1},\dots,\ba_{1,k}|0_k|\cdots|0_k)\\
    &({0},\ba_{2,2},\dots,\ba_{2,k}|\ba_{1,1},\dots,\ba_{1,k}|0_k|\cdots|0_k)\\
    &\hspace{0.5 cm}\vdots\\
    &(0_{r-1},\ba_{r,r},\dots,\ba_{r,k}|0_k|\cdots|\ba_{1,1},\dots,\ba_{1,k})\\
    &(0_k|0,\ba_{2,2},\dots,\ba_{2,k}|0_k|\cdots|0_k)\\
    &(0_k|0,0,\ba_{3,3},\dots,\ba_{3,k}|0,\ba_{2,2},\dots,\ba_{2,k}|\cdots|0_k)\\
    &\hspace{0.5 cm}\vdots\\
    &(0_k|0_{r-1},\ba_{r,r},\dots,\ba_{r,k}|0_k|\cdots|0,\ba_{2,2},\dots,\ba_{2,k})\\
    &\hspace{0.5 cm}\vdots\\
    &(0_k|\cdots|0_{r-2},\ba_{r-1,r-1},\dots,\ba_{r-1,k}|0_k)\\
    &(0_k|\cdots|0_{r-1},\ba_{r,r},\dots,\ba_{r,k}|0_{r-2},\ba_{r-1,r-1},\dots,\ba_{r-1,k})\\
    &(0_k|\cdots|0_k|0_{r-1},\ba_{r,r},\dots,\ba_{r,k}),
\end{aligned} $$
where $0_i$ denotes a string of $0$'s of length $i$. \\ To show that these vectors are contained in the left kernel of $M$, let us just consider
$$
(0_k|\cdots|0_{l},\ba_{l+1,l+1},\dots,\ba_{l+1,k}|0_{l-1},\ba_{l,l},\dots,\ba_{l,k}|\cdots|0_k)$$
and focus on the blocks $M_{l,l},M_{l+1,l},M_{l,l+1}$ and $M_{l+1,l+1}$. The other cases can be similarly treated. Then, we have
$$ M_{l,l}=
\begin{pmatrix}
    & 0_{l-1\times k-l}&\\
    -\ba_{l,l+1}&\dots&-\ba_{l,k}\\
    \ba_{l,l}&&\\
    &\ddots&\\
    &&\ba_{l,l}
\end{pmatrix}, \quad
M_{l+1,l}=
\begin{pmatrix}
    & 0_{l-1\times k-l}&\\
    -\ba_{l+1,l+1}&\dots&-\ba_{l+1,k}\\
    0&&\\
    &\ddots&\\
    &&0
    \end{pmatrix},
$$
$$ M_{l,l+1}=
\begin{pmatrix}
    & 0_{l\times k-l-1}&\\
    -\ba_{l,l+2}&\dots&-\ba_{l,k}\\
    \ba_{l,l+1}&&\\
    &\ddots&\\
    &&\ba_{l,l+1}
\end{pmatrix}, \quad
M_{l+1,l+1}=
\begin{pmatrix}
    & 0_{l\times k-l-1}&\\
    -\ba_{l+1,l+2}&\dots&-\ba_{l+1,k}\\
    \ba_{l+1,l+1}&&\\
    &\ddots&\\
    &&\ba_{l+1,l+1}
\end{pmatrix}.
$$
Thus, we get
\begin{gather*}
    (0_{l},\ba_{l+1,l+1},\dots,\ba_{l+1,k})\cdot M_{l,l}=(\ba_{l,l}\ba_{l+1,l+1},\dots,\ba_{l,l}\ba_{l+1,k}), \\
    (0_{l-1},\ba_{l,l},\dots,\ba_{l,k})\cdot M_{l+1,l}=(-\ba_{l,l}\ba_{l+1,l+1},\dots,-\ba_{l,l}\ba_{l+1,k})    
\end{gather*}
and
\begin{gather*}
    (0_{l},\ba_{l+1,l+1},\dots,\ba_{l+1,k})\cdot M_{l,l+1}=(-\ba_{l,l+2}\ba_{l+1,l+1}+\ba_{l,l+1}\ba_{l+1,l+2},\dots,-\ba_{l,k}\ba_{l+1,l+1}+\ba_{l,l+1}\ba_{l+1,k}), \\
    (0_{l-1},\ba_{l,l},\dots,\ba_{l,k})\cdot M_{l+1,l+1}=(\ba_{l,l+2}\ba_{l+1,l+1}-\ba_{l,l+1}\ba_{l+1,l+2},\dots,\ba_{l,k}\ba_{l+1,l+1}-\ba_{l,l+1}\ba_{l+1,k}).
\end{gather*}
Hence, these parts erase each other. All these $\binom{r+1} {2}$ vectors are linearly independent, therefore $rk-\binom{r+1}{2}=k-1+\dots+k-r$, implying that $\mathrm{rk}(M)=k-1+k-2+\dots+k-r$. In conclusion, we have shown that the space of solutions of \eqref{eq:systemwhole} has dimension $\binom{k-r}{2}$. \\ Note that when $r=k-1$ or $r=k$, we have that $M$ has maximal rank and therefore the only solution is the zero solution.
\end{proof}
\begin{remark}
    From the proof of Theorem \ref{th:MAIN}, we can note that computing the dimension of $ \mathcal{F}_s\cap \mathcal{I}(L_{\mU_G})$ is equivalent to solving a linear system with $k(n-k)$ equations involving $\binom{k}{2}$ variables.
\end{remark}
\begin{corollary}\label{cor:thmain}
    Let $\mathcal{C}\subseteq\F_{q^m}^n$  be a code of dimension $k$ with generator matrix $G=[\,I_k\,|\, X\,]$, where $X\in  \F_{q^m}^{k\times (n-k)}$. Let also $s\in\{1,\dots,m-1\}$  
    and $r=\mathrm{rk}(X^{[s]}-X)$. Then,
    \begin{equation*}
        h_{q^s+1}(\mC)\leq \binom{k+q^s}{q^s+1}-\binom{k-r}{2}.
    \end{equation*}
\end{corollary}
\begin{proof}
    Let $\mathcal{I}(L_{\mU_G})$ be the vanishing ideal of the linear set $L_{\mU_G}$. It suffices to notice that $$\dim_{\F_{q^m}}(\mI_{q^s+1}(L_{\mU_G}))\geq\dim_{\F_{q^m}}(\mathcal{I}(L_{\mU_G})\cap\mathcal{F}_s).$$ Then, by Theorem~\ref{th:MAIN} we obtain
    \begin{equation*}
        h_{q^s+1}(\mC)=\binom{k+q^s}{q^s+1}-\dim(\mI_{q^s+1}(L_{\mU_G}))\leq \binom{k+q^s}{q^s+1}-\binom{k-r}{2}.\qedhere
    \end{equation*}
\end{proof}

\noindent We are now in position to prove Theorem~\ref{theorem:codeq+1}.
\begin{proof}[Proof of Theorem~\ref{theorem:codeq+1}]
    Since $\mathrm{PG}(k-1,q)$ is contained in $L_{\mU_G}$, we have that $\mI(L_{\mU_G})_{q+1}\subseteq\mathcal{F}_1$, and therefore $\mI(L_{\mU_G})_{q+1}=\mathcal{I}(L_{\mU_G})\cap\mathcal{F}_1$. The result follows by applying Theorem~\ref{th:MAIN}. 
\end{proof}

\noindent From Theorem~\ref{th:MAIN} and Lemma~\ref{lm:rank1}, we immediately obtain the following characterization of generalized Gabidulin codes. 
\begin{corollary}
    Let $\mathcal{C}\subseteq\F_{q^m}^n$  be an MRD code of dimension $k$ with generator matrix $G=[\,I_k\,|\, X\,]$, where $X\in  \F_{q^m}^{k\times (n-k)}$, and let $s\in\{1,\dots,m-1\}$ with $\gcd(s, m) = 1$. Let also $\mathcal{I}(L_{\mU_G})$ be the vanishing ideal of the linear set $L_{\mU_G}$. Then, $\mC$ is a generalized Gabidulin code with parameter $s$ if and only if 
    $$\dim \mathcal{F}_s\cap \mathcal{I}(L_{\mU_G})=\binom{k-1}{2}.$$
\end{corollary}

\noindent Corollary \ref{cor:thmain} provides also a lower bound on the Castelnuovo-Mumford regularity of a code $\mC$. In particular, we obtain the following.
\begin{corollary}
      Let $\mathcal{C}\subseteq\F_{q^m}^n$  be a code of dimension $k$ with generator matrix $G=[\,I_k\,|\, X\,]$, where $X\in  \F_{q^m}^{k\times (n-k)}$.
      Let $s$ be the maximum integer in $\{1,\dots,m-1\}$
      such that 
      $$
      \binom{k+q^s}{q^s+1}-\binom{k-\mathrm{rk}(X^{[s]}-X)}{2}<|L_{\mU_G}|.
      $$
      Then, $q^s+1<r(\mC)$.
\end{corollary}
In this section we have proven that the $(q+1)$-th element of the $\F_q$-Hilbert sequence provides the same information as the dimension of $\Lambda_1(\mC)$.  However, the $\F_q$-Hilbert sequence and the sequence of dimensions of the $q$-sums, $\dim(\Lambda_i(\mC))$, are not the same invariant. There are cases in which two codes are distinguished by the $\F_q$-Hilbert sequence but not by the dimensions of the $q$-sums, and vice versa, as illustrated in the following examples.
\begin{example}
    Let $\mC_1$ and $\mC_2$ be two $[6,3]_{2^8/2}$ linear codes generated respectively by
    \begin{equation*}
        G_1=\begin{pmatrix}
            1&0&0&\alpha^{95}&\alpha^{173}&\alpha^{98}\\
            0&1&0&\alpha^{54}&\alpha^{218}&\alpha^{109}\\
            0&0&1&\alpha^{12}&\alpha^{98}&\alpha^{135}
        \end{pmatrix}\text{ and }G_2=\begin{pmatrix}
            1&0&0&\alpha^{8}&\alpha^{35}&\alpha^{75}\\
            0&1&0&\alpha^{250}&\alpha^{88}&\alpha^{163}\\
            0&0&1&\alpha^{51}&\alpha^{116}&\alpha^{141}
        \end{pmatrix},
    \end{equation*}
    where $\alpha^8+\alpha^4+\alpha^3+\alpha^2+1=0$. By simple computations, we can see that $\dim(\Lambda_1(\mC_1))=\dim(\Lambda_1(\mC_2))=6$. Since the dimension of the ambient space is $6$, we can immediately conclude that $\dim(\Lambda_i(\mC_1))=\dim(\Lambda_i(\mC_2))$, for all $i\geq1$. In particular, for any pair of distinct $\sigma_1,\sigma_2\in\mathrm{Gal}(\mathbb{F}_{2^8}/\mathbb{F}_2)$, $\dim(\sigma_1(\mC_1)\cap\sigma_2(\mC_1))=\dim(\sigma_1(\mC_2)\cap\sigma_2(\mC_2))=0$, implying that invariants based on the field automorphisms (as that in \cite{neri2020equivalence}) cannot be used to determine the nonequivalence of the two codes. \\ On the other hand, one can verify that the $\F_2$-Hilbert sequence of $\mC_1$ is \[\{1,3, 6, 10, 15, 21, 28, 35, 42, 49, 56, 62, 63,63,\dots\},\] while the $\F_2$-Hilbert sequence of $\mC_2$ is \[\{1,3, 6, 10, 15, 21, 28, 35, 42, 49, 56, 61, 63, 63,\dots\}.\] In particular, the two sequences coincide at all positions except for the third-to-last displayed here, which is $62$ in the first case and $61$ in the second. Therefore, $\mC_1$ and $\mC_2$ are not equivalent.
\end{example}

\begin{example}
 Let $\mC_1$ and $\mC_2$ be two $[4,2]_{2^4/2}$ linear codes generated respectively by
 \begin{equation*}
     G_1=\begin{pmatrix}
            1&0&0&\alpha^{12}\\
            0&1&\alpha^{14}&\alpha^{10}
        \end{pmatrix}\text{ and }G_2=\begin{pmatrix}
            1&0&\alpha&\alpha^{5}\\
            0&1&\alpha^{4}&\alpha^{10}
        \end{pmatrix},
 \end{equation*}
 where $\alpha^4+\alpha+1=0$. On one hand, one can verify that the $\F_2$-Hilbert sequences of both $\mC_1$ and $\mC_2$ are equal to $\{1,2,3,4,5,6,7,8,9,9,\dots\}$. On the other hand, for any pair of distinct $\sigma_1,\sigma_2\in\mathrm{Gal}(\mathbb{F}_{2^4}/\mathbb{F}_2)$, $\dim(\sigma_1(\mC_1)\cap\sigma_2(\mC_1))=0$, while $\dim(\sigma_1(\mC_2)\cap\sigma_2(\mC_2))=1$. In particular, we have that $4=\dim(\Lambda_1(\mC_1))\neq \dim(\Lambda_2(\mC))=3$, implying that the two codes are not equivalent.
\end{example}
\section{Some considerations on the number of ``zeros" of a {$(q^s+1)$}-form} \label{section:number of zeros of a form}
In the previous sections, we have shown that two equivalent linear rank-metric codes have the same $\F_q$-dimension sequence, which allows us to distinguish a Gabidulin code from a random one. At this point, it is natural to ask which properties of a linear code can be derived from its $\F_q$-dimension sequence. For instance, given a linear rank-metric code $\mC$ with $\F_q$-dimension sequence $\{h_i(\mC)\}_{i\geq0}$, we can easily find its dimension as $\dim_{\Fm}(\mC)=h_1(\mC)$. However, we are not able to determine the length of the code, as one can see in the following example.
\begin{example}
Let $\mC_1$ and $\mC_2$ be two $\F_4$-linear rank-metric codes generated respectively by
\begin{equation*}
    G_1=\begin{pmatrix}
        1&0&\alpha\\0&1&0
    \end{pmatrix}\text{ and }G_2=\begin{pmatrix}
        1&0&\alpha&0\\0&1&0&\alpha
    \end{pmatrix},
\end{equation*}
 where $\alpha$ is a primitive element of $\F_4$. Even though $\mathcal{U}_{G_1}\neq\mathcal{U}_{G_2}$, we have that $\PG(1,4)=L_{\mathcal{U}_{G_1}}=L_{\mathcal{U}_{G_2}}$. Therefore, $h_i(\mC_1)=h_i(\mC_2)$ for all $i\geq0$. However, $\mC_1$ has length $3$, while $\mC_2$ has length $4$.
\end{example}

\noindent Even though the length of a code is not determined by its $\F_q$-dimension sequence, we can still recover some partial information about it. 
\begin{lemma}
    Let $\mC$ be an $[n,k,d]_{q^m/q}$ code with generator matrix $G$. Then, for all $i\geq0$,
    \begin{equation*}
        n\geq\log_q((q-1)h_{i}(\mC)+1).
    \end{equation*}
\end{lemma}
\begin{proof}
    By Proposition \ref{proposition:dimseqalg}, $\lvert L_{\mathcal{U}_G}\rvert=h_{r(\mC)}(\mC)\geq h_i(\mC)$. Hence, since $q^n\geq|\,\mU_G\,|\geq (q-1)\lvert L_{\mathcal{U}_G}\rvert+1$, we immediately get the theorem.
\end{proof}

\noindent When the length of the code $\mC$ is given, its $\F_q$-dimension sequence allows to recover further information about it. 

\begin{remark}
When $n\leq m$, a code is MRD if and only if the linear set associated is scattered with respect to hyperplanes~\cite[Corollary 5.7]{sheekey2020rank}. Clearly this implies that the linear set is scattered, hence its cardinality is equal to $\frac{q^n-1}{q-1}$. Therefore, if $h_{r(\mC)}(\mC)<\frac{q^n-1}{q-1}$, the code $\mC$ can not be MRD. \\ Moreover, if $k=2$, hyperplanes are points, implying that being scattered with respect to hyperplanes coincides with being scattered. In this case, $h_{r(\mC)}(\mC)=\frac{q^n-1}{q-1}$ if and only if $\mC$ is MRD.
\end{remark}

\noindent Even when only a subsequence of the $\F_q$-dimension sequence is known, it is still possible to obtain information on $h_{r(\mC)}(\mC)$, or equivalently on $\lvert L_{\mathcal{U}_G}\rvert$. Since the $\F_q$-dimension sequence is nondecreasing, it immediately provides a lower bound on $h_{r(\mC)}(\mC)$. On the other side, $h_i(\mC)$ allows us to determine $\dim(\mI(L_{\mathcal{U}_G})_i)$, and any upper bound on the variety described by the ideal generated by $\mI(L_{\mathcal{U}_G})_i$ gives an upper bound on $\lvert L_{\mathcal{U}_G}\rvert$. In particular, for $i=q^s+1$, we have that $$\lvert L_{\mathcal{U}_G}\rvert\leq \lvert{\rm V}(\mathcal{F}_s\cap\mI(L_{\mathcal{U}_G})) \rvert,$$ where ${\rm V}(\mathcal{F}_s\cap\mI(L_{\mathcal{U}_G}))$ is the variety of the ideal generated by $\mathcal{F}_s\cap\mI(L_{\mathcal{U}_G})$. Since for each point in $L_{\mathcal{U}_G}$ there exists an $\alpha\in\Fm^k$ such that $p(\mathcal{V}_{\alpha})=0$ for all $p\in\mathcal{F}_s\cap\mI(L_{\mathcal{U}_G})$, we can consider only the zeros of this type. Therefore, we obtain
\begin{equation*}
    \lvert L_{\mathcal{U}_G}\rvert\leq \frac{\lvert\{\mathcal{V}_{\alpha}\st \alpha\in\Fm^k\text{ and }p(\mathcal{V}_{\alpha})=0\text{ for all }p\in\mathcal{F}_s\cap\mI(L_{\mathcal{U}_G})\}\rvert q^k-1}{q-1}.
\end{equation*}
Hence, to derive an upper bound on $h_{r(\mC)}(\mC)$ it is sufficient to upper bound the number of $\mathcal{V}_{\alpha}$'s vanishing on a subset $\mathcal{S}$ of $\mathcal{F}_s$. In the last part of the section, we focus on the case where $\lvert\mathcal{S}\rvert=1$. Note that the question about the behavior of forms over linear sets has its own mathematical interest.

We recall that $s\in\{1,\dots,m-1\}$ and $\gcd(s,m)=1$.
\begin{lemma}\label{lm:coset}
        Let $\alpha,\beta\in \F_{q^m}^k$. Then, $\mV_\alpha=\mV_\beta$ if and only if $\alpha^{[s]}-\alpha=\beta^{[s]}-\beta$.
\end{lemma}
\begin{proof}
       $\mV_\alpha=\mV_\beta$ directly implies that $\beta=\alpha+v$ for some $v\in\F_q^k$. Therefore, we immediately get that $\beta^{[s]}-\beta=(\alpha+v)^{[s]}-(\alpha+v)=\alpha^{[s]}-\alpha$. \\ On the other side, $\alpha^{[s]}-\alpha=\beta^{[s]}-\beta$ implies that $(\alpha-\beta)^{[s]}=\alpha^{[s]}-\beta^{[s]}=\alpha-\beta$. Since $\gcd(s,m)=1$, from basic Galois theory, we have that $\alpha-\beta\in\F_q^k$. Hence, $\mV_\alpha=\mV_\beta$.
\end{proof}
\noindent The previous lemma implies that the number of sets $\mathcal{V}_\alpha$ on which a form $p$ vanishes is given by the cardinality of $$\bar{\mathcal{Z}}_p:=\{\alpha^{[s]}-\alpha\st\alpha\in\Fm^k, \ p(\mV_\alpha)=0\}.$$
\noindent In addition, by Theorem \ref{th:MAIN}, we can conclude the following.

\begin{lemma}\label{lm:zeroform}
    Let $\alpha_1,\dots,\alpha_{k-1}\in\F_{q^m}^k$ be such that $\bar\alpha_i=\alpha_i^{[s]}-\alpha_i$, for $i\in\{1,\dots,k-1\}$, are linearly independent. If $p\in \mathcal{F}_s$ vanishes on all the sets $\mV_{\alpha_i}$, then $p\equiv 0$.
\end{lemma}
\begin{proof}
Take $\alpha_1,\dots,\alpha_{k-1}\in\F_{q^m}^k$ to be the column of a $(k-1)\times k$ matrix $X$ over $\F_{q^m}$. Now, $\mathrm{rk}(X^{[s]}-X) = k-1$. Hence, $\dim \mathcal{F}_s\cap \mathcal{I}(L_{\mU_G})=0$, where $G=[\ I_k \mid X \ ]$, by Theorem \ref{th:MAIN}. Then $p$ should be the zero form. 
\end{proof}

\begin{theorem}\label{th:zeros}
    Let $p\in\mathcal{F}_s$ be a nonzero $(q^s+1)$-form. If $\dim_{\F_{q^m}}(\langle \bar{\mathcal{Z}}_p\rangle_{\F_{q^m}})=r$, then there exist $\alpha_1,\dots,\alpha_r\in\Fm^k$ such that  $$q^{\sum_{i=1}^r\dim_{\F_q}(\mathcal{H}_{\alpha_i})}\le |\bar{\mathcal{Z}}_p|\le q^{r(m-1)},$$ where $\mathcal{H}_{\alpha_i}=\big\{v\in\Fm\st\mathrm{Tr}_{q^m/q}\big((\alpha_{i,j}^{[s]}-\alpha_{i,j})v\big)=0\text{ for all } j\in\{1,\dots,k\}\big\}$.  In particular, the maximum number of sets $\mathcal{V}_\alpha$ on which  a $(q^s+1)$-form in $\mathcal{F}_s\setminus\{0\}$ can vanish is $q^{(k-2)(m-1)}$.
\end{theorem}
\begin{proof}
    Let $\alpha_1,\dots,\alpha_r\in\F_{q^m}^k$ be such that $\bar{\mathcal{Z}}_p\subseteq\langle\bar\alpha_1,\dots,\bar \alpha_r\rangle_{\Fm}$, where $\bar{\alpha}_i=\alpha_i^{[s]}-\alpha_i$.
    Let $\mathcal{T}_0:=\{\bar{\beta}\,:\,\beta\in\Fm^k\}$ be the set of vectors in $\Fm^k$ with entries having trace equal to $0$ over $\Fq$. Then, from Lemma \ref{lm:aU}, we obtain that
    $$\bar{\mathcal{Z}}_p=\langle\bar \alpha_1,\dots,\bar\alpha_r \rangle_{\Fm}\cap \mathcal{T}_0.$$
    Since $\bar \alpha_1,\dots,\bar\alpha_r$ are linearly independent over $\Fm$, there exist $r$ entries such that the projection of $\langle\bar \alpha_1,\dots,\bar\alpha_r \rangle_{\Fm}$ over these entries is a bijection between $\langle\bar \alpha_1,\dots,\bar\alpha_r \rangle_{\Fm}$ and $\Fm^r$. Hence, the image of $\langle\bar \alpha_1,\dots,\bar\alpha_r \rangle_{\Fm}\cap \mathcal{T}_0$ under this projection is contained in the set of vectors in $\Fm^r$ whose entries have trace equal to zero, that has cardinality equal to $q^{r(m-1)}$. On the other side, since $\bar{\mathcal{Z}}_p$ is an $\F_q$-linear space, every element of the form $\sum_{i=1}^rv_i\bar\alpha_i$, for $v_i\in\mathcal{H}_{\alpha_i}$, belongs to $\mathcal{T}_0$, hence to $\bar{\mathcal{Z}}_p$. Thus, we obtain the lower bound on the cardinality of $\bar{\mathcal{Z}}_p$. Finally, by Lemma \ref{lm:zeroform}, we get that the cardinality of $\bar{\mathcal{Z}}_p$ can be at most $q^{(k-2)(m-1)}$.
\end{proof} \smallskip

All the arguments above can be generalized to the case $\gcd(s,m)=\delta>1$. \\ In this case, let us consider $\tilde\mV_\alpha:=\alpha+\mathbb{F}_{q^\delta}^k$. Then, $\tilde\mV_\alpha$ can be seen as the disjoint union of $q^{\delta-1}$ cosets of $\Fq^k$. That is, there exist $v_1,\dots,v_{q^{\delta-1}}\in\mathbb{F}_{q^\delta}^k$ such that \[\tilde\mV_\alpha=\bigcup_{i=1}^{q^{\delta-1}}\mV_{\alpha+v_i}.\] As in Lemma \ref{lm:coset}, we have that $\tilde\mV_\alpha=\tilde\mV_\beta$ if and only if $\alpha^{[s]}-\alpha=\beta^{[s]}-\beta$. Moreover, by Lemma~\ref{lm:aU}, if $p$ is a $(q^s+1)$-form that vanishes over $\mV_{\alpha+v_i}$ for some $i\in\{1,\dots,q^{\delta-1}\}$, then it vanishes over $\mV_{\alpha+v_j}$ for any $j\in\{1,\dots,q^{\delta-1}\}$. Therefore, the number of sets $\mathcal{V}_\alpha$ on which a form $p$ vanishes is given by $q^{\delta-1}$ times the cardinality of $$\tilde{\mathcal{Z}}_p:=\{\alpha^{[s]}-\alpha\st\alpha\in\Fm^k, \ p(\tilde\mV_\alpha)=0\}.$$ Hence, by similar arguments as those in Theorem \ref{th:zeros}, and considering the trace function $\mathrm{Tr}_{q^m/q^\delta}:\Fm\to \mathbb{F}_{q^\delta}$, we get the following.

\begin{theorem}\label{th:zeros2}
    Let $s\in\{1,...,m-1\}$, with $\gcd(s,m)=\delta$. Let $p\in\mathcal{F}_s$ be a nonzero $(q^s+1)$-form. If $\dim_{\F_{q^m}}(\langle \bar{\mathcal{Z}}_p\rangle_{\F_{q^m}})=r$, then there exist $\alpha_1,\dots,\alpha_r\in\Fm^k$ such that, for all $i\in\{1,\dots,r\}$, $$(q^\delta)^{\sum_{i=1}^r\dim_{\F_{q^\delta}}(\mathcal{H}_{\alpha_i})}\le |\tilde{\mathcal{Z}}_p|\le q^{r(m-\delta)},$$ where $\mathcal{H}_{\alpha_i}=\big\{v\in\Fm\st\mathrm{Tr}_{q^m/q^\delta}\big((\alpha_{i,j}^{[s]}-\alpha_{i,j})v\big)=0\text{ for all } j\in\{1,\dots,k\}\big\}$.  In particular, the maximum number of sets $\mathcal{V}_\alpha$ on which  a $(q^s+1)$-form in $\mathcal{F}_s\setminus\{0\}$ can vanish is $q^{\delta-1}q^{(k-2)(m-\delta)}$.
\end{theorem}
\noindent We conclude noting that the upper bound in Theorem \ref{th:zeros2} is tight. Indeed, we can meet it for instance by considering $\gamma e_1,\dots,\gamma e_{k-2}$, where $\gamma\in\Fm\setminus\F_{q^\delta}$ and $e_i$ is the $i$-th vector of the canonical base of $\Fm^k$.

\noindent The situation when $|\mathcal{S}|>1$ is way more complicated and needs further investigations.

\bigskip

\section*{Acknowledgments}
The authors of this paper would like to thank Alain Couvreur, Elisa Gorla, and Ferdinando Zullo for their insightful feedback and suggestions. They would also like to thank the reviewers for their thoughtful and constructive comments, which helped us improve the quality and clarity of the manuscript.

\bibliographystyle{abbrv}
\bibliography{biblio.bib}
\end{document}